\documentclass[12pt]{article}

\usepackage{epsfig}

\usepackage{amssymb}

\def\be{\begin{equation}}
\def\ee{\end{equation}}
\def\ba{\begin{array}{c}}
\def\ea{\end{array}}

\def\ben{$$}
\def\een{$$}

\newcommand{\bea}{\begin{eqnarray}}
\newcommand{\eea}{\end{eqnarray}}

\newcommand{\pbr}{\prec\!}
\newcommand{\pkt}{\!\!\succ\,\,}
\newcommand{\kt}{\rangle}
\newcommand{\br}{\langle}

\newtheorem{thm}{Theorem}

\newenvironment{proof}{\noindent {\bf Proof.}}{\hfill$\square$\vspace{3mm}\endtrivlist}

\begin{document}

\begin{center}

{\Large \bf

Non-Hermitian interaction representation and its use in relativistic
quantum mechanics

  }

\vspace{9mm}

{Miloslav Znojil}

\vspace{9mm}

Nuclear Physics Institute CAS, 250 68 \v{R}e\v{z}, Czech Republic

{znojil@ujf.cas.cz}

{http://gemma.ujf.cas.cz/\~{}znojil/}

\end{center}

\section*{Abstract}

The textbook interaction-picture formulation of quantum mechanics is
extended to cover the unitarily evolving systems in which the
Hermiticity of the observables is guaranteed via an {\em ad hoc\,}
amendment of the inner product in Hilbert space. These systems are
sampled by the Klein-Gordon equation with a space- and
time-dependent mass term.

\section*{Keywords}

relativistic quantum mechanics; unitary evolution; non-Hermitian
interaction picture; state vectors; Schr\"{o}dinger-type equations;
observables; Heisenberg-type equations; Klein-Gordon particles;
instant- and position-dependent mass;

\newpage

\section{Introduction \label{oh} }

Quantum mechanics formulated in the Dirac's {\it alias\,}
interaction picture offers a fairly universal representation of the
unitarily evolving quantum systems as well as a useful methodical
bridge between the most economical Schr\"{o}dinger picture  and its
slightly older, more intuitive Heisenberg-picture  alternative
\cite{Messiah}. Unfortunately, in the eyes of mathematical
physicists the overall image of the interaction picture approach is
partially damaged by its frequent use in the areas prohibited by the
Haag's theorem \cite{Haag}. The applicability of the recipe was
rigorously disproved in multiple models of quantum fields
\cite{Hall}. Nevertheless, one may feel surprised that the review
paper ``Nine formulations of quantum mechanics'' \cite{nine} did not
mention the interaction picture even without moving into the quantum
field theory. This makes an impression of a misunderstanding because
the authors of the review admit, at the same time, the usefulness of
having alternative approaches. They emphasized that from a purely
technical point of view ``no formulation produces a royal road to
quantum mechanics'' \cite{nine}.

In the technical setting we are witnessing a remarkable new progress
during the last few years. The conventional menu of pictures was
enriched by the descriptions of the unitarily evolving quantum
systems in which the optimal representations of certain observables
appear non-Hermitian (cf., e.g.,~\cite{Dyson,Geyer}). In particular,
multiple specific merits of the non-Hermitian evolution equations
were found to follow from the use of the generalized stationary
non-Hermitian Scjr\"{o}dinger picture \cite{Carl,ali,book}. In this
amended framework people started using, with great success and
impact, innovated stationary non-Hermitian Hamiltonians $H \neq
H^\dagger$ with real spectra and exhibiting certain  unusual
features like parity-time symmetry, etc (see the recent review of
history \cite{MZbook} and a few more detailed comments in section
{\ref{seci1}}  below).

The scope of the early extensions of the non-Hermitian formalism to
non-stationary dynamical scenarios \cite{Fringlas,PLB} was strongly
limited by the traditional belief that every consistent form of the
non-stationary generator of the evolution of wave functions (let us
denote this generator by a dedicated symbol $G(t)$) must be
observable. This belief proved unexpectedly deeply rooted (cf. also
the early extensive discussion \cite{web} or Theorem 2 in
\cite{ali}). Once this belief was reclassified as artificial and
redundant, the non-Hermitian formalism was immediately generalized
to non-stationary systems \cite{timedep,SIGMA}. Its rather
complicated technical nature inspired its further simplifications
\cite{IJTP,Luiz} and, finally, a reduction to the special stationary
non-Hermitian Heisenberg picture \cite{Heisenberg} with certain
specific practical merits as discussed, e.g., in
Refs.~\cite{symmetry,cinani}.

The difficulty of a return to the full-fledged non-stationary
formalism of the non-Hermitian interaction picture remained
perceived as a challenge \cite{Wang,Maamache,FringMou}. In our
present paper we shall return to the subject, with the emphasis on
the building of a methodical bridge between the alternative
non-Hermitian formulations of quantum mechanics. We shall summarize
the recent developments in this direction, formulate a new
consistent version of the non-Hermitian interaction picture and we
shall illustrate its applicability via a realistic, relativistic
quantum-mechanical model.

The material will be organized as follows. In sections \ref{seci1}
and \ref{secwe} we shall describe the necessary methodical
preliminaries. We shall recall the notation of review~\cite{SIGMA}
and we shall explain that the phenomenological models with both
stationary and non-stationary non-Hermitian generators of wave
functions can be perceived as unitary and phenomenologically
admissible. In sections \ref{seci2}, \ref{secweb}, \ref{seci3b} and
\ref{cwe} the formalism will be made complete while in
section~\ref{seci3} the theory will be applied to the Klein-Gordon
model, with the mass term being both space- and time-dependent. In
the last two sections \ref{seci4} and \ref{seci5} we shall discuss
some of the possible further innovative consequences of our approach
in the broader context of physics.

\section{\label{seci1}Stationary quantum systems in the
non-Hermitian Schr\"{o}dinger picture}

The existing nontrivial applications of the non-Hermitian
representations of unitary quantum systems are almost exclusively
constructed in the Schr\"{o}dinger picture (SP) framework. The main
reason lies in the enhanced importance of the mathematical economy
of the non-Hermitian formalism, much more rarely achieved in the
alternative Heisenberg picture (HP, \cite{Heisenberg,cinani}).

\subsection{Conventional Hermitian observables: Limits of applicability}

An arbitrary quantum system ${\cal S}$ may be studied in the most
common Schr\"{o}dinger picture of textbooks, in principle at least.
The system is represented, in pure state, by an element of a
preselected Hilbert space ${\cal H}^{(textbook)}$, i.e., by a
ket-vector $|\psi^{(SP)}(t)\pkt \in {\cal H}^{(textbook)}$ which is
time-dependent. After one prepares this ket vector at an initial
time $t_i=0$, the task is to predict the result of a measurement
performed at a future time $t_f>0$. This prediction is
probabilistic. Whenever we are interested in the measurement of an
observable represented by a stationary self-adjoint operator
$\mathfrak{q}_{(SP)}$ (or, in general, by the variable
$\mathfrak{q}_{(SP)}(t)$), the prediction is prescribed in terms of
matrix element
 \be
 \pbr \psi^{(SP)}(t_f)|\mathfrak{q}_{(SP)}(t_f) |\psi^{(SP)}(t_f)
 \pkt\,.
 \label{MEAS}
 \ee
The experimentally relevant information about the unitary evolution
of the system ${\cal S}$ is all carried by the wave-ket solutions of
Schr\"{o}dinger equation
 \be
 {\rm i} \frac{\partial}{\partial t} \,|\psi(t)\pkt =
 \mathfrak{h}_{(SP)}\,|\psi(t)\pkt\,,
 \ \ \ \ \
 \ \ \ \ \ |\psi(t)\pkt \in {\cal H}^{(textbook)}\,.
 \label{SET}
 \ee
The unitarity of the evolution is equivalent, due to the Stone's
theorem \cite{Stone}, to the self-adjoint nature of the system's
Hamiltonian, $\mathfrak{h}_{(SP)}=\mathfrak{h}_{(SP)}^\dagger\,$ in
${\cal H}^{(textbook)}$.

For the stationary self-adjoint Hamiltonians
$\mathfrak{h}_{(SP)}\neq \mathfrak{h}_{(SP)}(t)$ the solution of
Eq.~(\ref{SET}) may proceed via their diagonalization. In realistic
systems ${\cal S}$ even this diagonalization may happen to be a
formidable, practically next-to-impossible task. One of its
efficient {\em simplifications} is provided by a non-Hermitian
modification of Schr\"{o}dinger picture. This approach, usually
attributed to Dyson \cite{Dyson}, found a number of successful
applications in nuclear physics \cite{Geyer}. Its key features will
be described in the next paragraph.

\subsection{The change of Hilbert space using stationary
Dyson map}

Dyson \cite{Dyson} revealed that in many-fermion quantum systems
${\cal S}_{(Dyson)}$ one of the key technical obstacles lies in an
extremely unfriendly nature of the underlying conventional fermionic
Hilbert space ${\cal H}^{(textbook)}$ as well as in an equally
unfriendly nature of the eigenstates $|\psi_n\pkt$ of the system's
stationary self-adjoint Hamiltonians $\mathfrak{h}_{(SP)}$. As long
as at least a part of the difficulty lies in the fermion-fermion
correlations, Dyson proposed that these correlations might be
separated and simulated via an {\it ad hoc\,} operator,
 \be
 |\psi_n^{(correlated)}\pkt=\Omega_{(Dyson)}|\psi_n^{(simplified)}\kt
 \,.
 \label{lumapo}
 \ee
The best candidates for the correlation operators $\Omega_{(Dyson)}$
appeared to be non-unitary,
 \be
  \Omega_{(Dyson)} \neq
\Omega_{(Dyson)}(t)],,\ \ \ \ \ \
 \Omega^\dagger_{(Dyson)}\Omega_{(Dyson)}=\Theta_{(stationary)}
  \neq I\,.
 \label{nontri}
 \ee
Dyson decided to make their ``trial and error'' choice
$n-$independent. Thus, he replaced Eq.~(\ref{lumapo}) by an overall
ansatz
 \be
 |\psi^{(SP)}(t)\pkt=\Omega_{(Dyson)}|\psi^{(Dyson)}(t)\kt
 \,,
 \ \ \ \ \ \ |\psi^{(Dyson)}(t)\kt \in {\cal H}^{(friendlier)}\,.
 \label{mapo}
 \ee
in which the map $\Omega_{(Dyson)}\neq \Omega_{(Dyson)}(t)$ from the
simpler Hilbert space remains stationary.

In the numerous applications of such an ansatz in nuclear physics
\cite{Geyer} the mapping was reinterpreted as a replacement of the
unfriendly initial fermionic Hilbert space ${\cal H}^{(textbook)}$
by a more suitable {\em bosonic} Hilbert space ${\cal
H}^{(friendlier)}$. In the majority of these applications the
mapping was kept stationary but, otherwise, flexible, non-unitary.
The stationarity assumption implied that the Schr\"{o}dinger
Eq.~(\ref{SET}) in ${\cal H}^{(textbook)}$ gets replaced by its
``interacting boson model'' avatar
 \be
 {\rm i} \frac{\partial}{\partial t} \,|\psi^{(Dyson)}(t)\kt
 =H_{(Dyson)}\,|\psi^{(Dyson)}(t)\kt\,,
 \ \ \ \ \ \
 |\psi^{(Dyson)}(t)\kt \in {\cal H}^{(friendlier)}\,.
 \label{SETdys}
 \ee
It lives in the friendlier space and contains a {\em simplified\,}
isospectral Hamiltonian,
 \be
 H_{(Dyson)}= \Omega_{(Dyson)}^{(-1)}
 \mathfrak{h}_{(SP)}\Omega_{(Dyson)}\,.
 \label{dag}
 \ee
Needless to add that the variational determination of spectra via
$H_{(Dyson)}$ proved superior in a number of realistic calculations
\cite{Geyer}.

\section{The Stone's theorem revisited\label{secwe}}

In a brief conceptual detour let us now address one of the apparent
paradoxes encountered after the use of non-unitary maps
(\ref{nontri}).

\subsection{The Dyson-proposed return to Hermiticity}

The replacement (\ref{dag}) of the known and realistic Hamiltonian
$\mathfrak{h}_{(SP)}$ defined in ${\cal H}^{(textbook)}$ by its
simpler isospectral partner $H_{(Dyson)}$ acting in ${\cal
H}^{(friendlier)}$ is accompanied by the loss of the physical status
of the Hilbert space caused by the non-unitarity of the map. As long
as the product (\ref{nontri}) called metric remains
time-independent, its emergence introduces a comparatively marginal
technical complication. It suffices to convert the auxiliary,
unphysical Hilbert space ${\cal H}^{(friendlier)}$ into its
unitarily non-equivalent partner, i.e., into another Hilbert space
${\cal H}^{(standard)}$. The latter space re-acquires the
physical-state status because it becomes, by construction, unitarily
equivalent to ${\cal H}^{(textbook)}$.

The conversion is easy. One merely keeps the linear space ${\cal V}$
of the ket-vectors $|\psi\kt$ unchanged and realizes the transition
to the desired ${\cal H}^{(standard)}$ by the mere {\it ad hoc\,}
amelioration of the bra vectors,
 \be
 \br \psi | \ \to \  \br \psi | \Theta_{(stationary)}
  \ \equiv \ \br \psi_\Theta |
 \ \ \ \ \ {\rm for} \ \ \ \ \
 {\cal H}^{(friendlier)} \to {\cal H}^{(standard)}
 \,.
 \label{amelio}
 \ee
The new bra-vectors are metric-dependent. In the terminology of
functional analysis we just changed the definition of the dual {\it
alias} bra-vector space of the linear functionals, ${\cal V}' \to
{\cal V}'_\Theta\,$. Thus, we replaced the unphysical Dirac's
bra-ket inner product $ \br \psi | \chi\kt $ by its physical
standard-space amendment,
 \be
 \br \psi
 |\chi \kt \to \br \psi|\Theta |\chi \kt \ \equiv \ \br \psi_\Theta
 |\chi \kt\,.
  \ee
Schematically, the situation is depicted in Fig.~\ref{picee3wwww}.

\begin{figure}[h]                     
\begin{center}                         
\epsfig{file=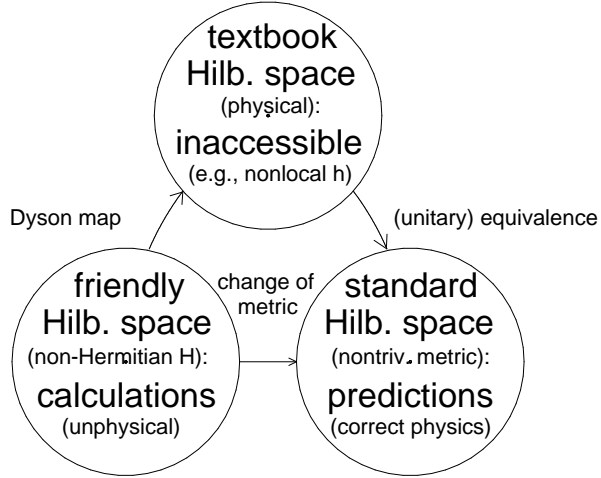,angle=270,width=0.7\textwidth}
\end{center}                         
\vspace{-2mm}\caption{Mutual relations of the three Hilbert spaces
needed in quantum theory using non-Hermitian operators (stationary
case).
 \label{picee3wwww}}
\end{figure}

In Fig.~\ref{picee3wwww} we emphasize the strict equivalence of the
phenomenological predictions which are, in principle, available in
both of the physical Hilbert spaces ${\cal H}^{(textbook)}$ and
${\cal H}^{(standard)}$. This explains why the theory does not
contradict the Stone's theorem \cite{Stone}. Indeed, this theorem
states that the unitarity of the evolution {\em requires} the
self-adjointness of the generator. Obviously, by construction, such
a condition is satisfied in {\em both} of the physical Hilbert
spaces, viz., in ${\cal H}^{(textbook)}$ (where we have
$\mathfrak{h}=\mathfrak{h}^\dagger$, as we must have) as well as in
${\cal H}^{(standard)}$ (where we only have to understand that the
self-adjointness is usually represented by formula
$H=H^\ddagger:=\Theta^{-1}H^\dagger\Theta$, i.e., by its
representation in auxiliary  ${\cal H}^{(friendlier)}$).

\subsection{The predictions of measurements revisited}

In practice, we are always performing all calculations in the
auxiliary space ${\cal H}^{(friendlier)}$. Whenever needed, we only
perform the easy transition to ${\cal H}^{(standard)}$ by means of
the introduction of the non-trivial {\it ad hoc} metric operator
$\Theta \neq I$. The unitarity of the evolution is equally
guaranteed in both of the representations of the quantum system in
question, therefore.

The results of measurements find their correct interpretation in
either one of the two, mutually unitarily equivalent physical
Hilbert spaces ${\cal H}^{(standard)}$ and ${\cal H}^{(textbook)}$.
The main advantage of the resulting triple-representation formalism
is that the majority of the constructive considerations may be, in
the auxiliary Hilbert space ${\cal H}^{(friendlier)}$, maximally
simplified. Only the definitions of the observables
$\mathfrak{q}_{(SP)}$ are usually deduced inside ${\cal
H}^{(textbook)}$. Thus, one has to use an analogue of
Eq.~(\ref{dag}) and one has to pull the operators up to the
auxiliary Hilbert space ${\cal H}^{(friendlier)}$. This yields their
non-Hermitian representations in ${\cal H}^{(friendlier)}$,
 \be
 Q_{(Dyson)}= \Omega_{(Dyson)}^{(-1)}
 \mathfrak{q}_{(SP)}\Omega_{(Dyson)} \neq Q_{(Dyson)}^\dagger\,.
 \label{dadag}
 \ee
In the case of the observable of position this transformation of
operators other than Hamiltonians was sampled in \cite{Batal}. In
all of the similar cases the experiment-predicting formula
(\ref{MEAS}) may be replaced by its --  friendlier -- alternative
since
 \be
 \pbr \psi^{(SP)}(t_f)|\mathfrak{q}_{(SP)}(t_f) |\psi^{(SP)}(t_f)
 \pkt=
 \br \psi^{(Dyson)}_\Theta(t_f)|Q_{(Dyson)}(t_f)
 |\psi^{(Dyson)}(t_f)
 \kt\,.
 \label{dysMEAS}
  \ee
The pragmatic appeal of Hamiltonians $H_{(Dyson)}$ and of the other
non-Hermitian observables $Q_{(Dyson)}$ results from an implicit
assumption that the use of the respective mappings (\ref{dag}) and
(\ref{dadag}) {\em simplifies} the evaluation of matrix
elements~(\ref{dysMEAS}). Otherwise, the use of the non-unitary map
of Eq.~(\ref{mapo}) and of the non-Hermitian version of
Schr\"{o}dinger equation would not be sufficiently well motivated.

\subsection{\label{ytseci2}Best known special case: Stationary
${\cal PT}-$symmetric systems}

In applications, the stationary maps $\Omega_{(Dyson)}$ are sought
in a trial-and-error manner \cite{Dyson,Geyer}. Bender with
Boettcher \cite{BB} innovated the strategy and proposed an inversion
of the Dyson's flowchart. This inspired a number of innovative
studies of Eq.~(\ref{SETdys}). All of them were based on an {\em
input\,} choice of a tentative stationary candidate $H_{(Dyson)}\neq
H_{(Dyson)}^\dagger$ for the Hamiltonian. These toy-model operators
were mostly chosen Krein-space self-adjoint {\it alias} ${\cal
PT}-$symmetric,
 \be
 H_{(Dyson)} {\cal PT}= {\cal PT} H_{(Dyson)}\,,
 \label{rel}
 \ee
with ${\cal PT}$ meaning parity-time-reflection (see an extensive
summary of these developments in deeply physics-oriented review
paper \cite{Carl}).

In the resulting narrower, ${\cal PT}-$symmetric quantum mechanics
the main task lies again in the ultimate evaluation of experimental
predictions (\ref{dysMEAS}). For this purpose people usually
reconstruct the observables-Hermitizing Hilbert-space metric
$\Theta=\Theta_{(stationary)}$ from the Hamiltonian, i.e., as a
solution of the time-independent relation
 \be
  H^\dagger \Theta=
  \Theta^{}H_{}
 \,,
 \ \ \ \
 \Theta=\Omega^\dagger\Omega\,
 \label{daobs}
 \ee
which follows from the self-adjointness constraint
$\mathfrak{h}=\mathfrak{h}^\dagger$ and from definition (\ref{dag}).

People often introduce also another observable $Q$. The task of
making the theory fully consistent is then more complicated
\cite{arabky}. In the light of Eq.~(\ref{dadag}) the same metric
must satisfy also the second equation
 \be
 Q^\dagger_{(Dyson)} \Theta
 =\Theta Q_{(Dyson)}\,.
 \label{die}
 \ee
It is worth adding that physicists often need the factor map
$\Omega_{(Dyson)}$ as well. This is another ambiguous task in
general. For the sake of simplicity this ambiguity is usually
ignored and the special, unique self-adjoint
$\Omega_{(special)}=\sqrt{\Theta}$ is used \cite{ali}.

\section{\label{seci2}
Non-stationary quantum systems}


In a way explained in the influential 2007 letter \cite{PLB} the
possibilities of an innovative, time-dependent choice of the
Hilbert-space mapping in ansatz (\ref{mapo}) are rather restricted.
The reasons may be found discussed in review paper \cite{ali} where
it has been proved that in an extended, non-stationary theory with
time-dependent metric  one has the choice between the loss of the
unitarity of the evolution and the loss of the observability of the
Hamiltonian operator. In our more or less immediate reaction
\cite{which} to the latter no-go theorem we pointed out that the
puzzle is purely terminological. The observability of the generator
of the evolution of the state-ket-vectors in Schr\"{o}dinger
Eq.~(\ref{SETdys}) is not necessary~\cite{Heisenberg,FringMou}. In
the non-stationary non-Hermitian cases it is possible to admit that
the generator $G(t)$  {\em differs} from the observable Hamiltonian
$H_{(Dyson)}(t)$ of Eq.~(\ref{dag}) by a new, Coriolis-force
component $\Sigma(t)$ of a purely kinematical origin (to be defined
by Eq.~(\ref{defsig}) in the next paragraph). Thus, one can speak
about a non-Hermitian interaction picture (NIP).


\subsection{Schr\"{o}dinger equation for ket
vectors\label{neramo}}

After one accepts the non-stationary quantum-evolution scenario, the
theory becomes complicated. Multiple new open questions emerge. They
concern not only the NIP theory itself (e.g., the perception of the
adiabatic approximation \cite{Bila,Stefan} or of the so called
geometric phases \cite{Wang,Maamache}) but also its practical
applications, say, in the analysis of the observability-breaking
processes \cite{Luiz,BBufab} or in an efficient elimination of the
time dependence from Schr\"{o}dinger equations \cite{IJTP,FrFrith}.

In a brief reminder of the time-dependent version of quantum theory
in its generalized, non-Hermitian and non-stationary
triple-Hilbert-space formulation of~Refs.~\cite{timedep,SIGMA} let
us abbreviate or drop the subscripts and superscripts. The
presentation of the theory may then start from a non-stationary
version of ansatz~(\ref{mapo}),
 \be
 |\psi^{}(t)\pkt=\Omega(t)|\psi^{}(t)\kt
 \in {\cal H}^{(T)}\,,
 \ \ \ \ \ \ |\psi^{}(t)\kt \in {\cal H}^{(F)}\,.
 \label{mapoge}
 \ee
The insertion in the original Schr\"{o}dinger Eq.~(\ref{SET}) does
not lead to Eq.~(\ref{SETdys}) but to its non-stationary update
 \be
 {\rm i} \frac{\partial}{\partial t}
 \,|\psi^{}(t)\kt=G_{}(t)\,|\psi^{}(t)\kt\,
  \label{SEFip}
 \ee
to be solved in ${\cal H}^{(F)}$, with the ``unobservable
Hamiltonian''
 \be
 G(t) = H(t) - \Sigma(t)
 \,
  \label{2.3}
 \ee
containing the ``observable Hamiltonian''
 \be
 H_{}(t)= \Omega_{}^{(-1)}(t)
 \mathfrak{h}_{(SP)}(t)\Omega_{}(t)
 \label{udagobs}
 \ee
and the ``Coriolis-force Hamiltonian''
 \be
 \Sigma(t)={\rm i} \Omega^{-1}(t)\dot{\Omega}(t)\,,
 \ \ \ \ \
 \dot{\Omega}(t)=
 \frac{d}{dt} \,
 \Omega(t)\,.
 \label{defsig}
 \ee
Operator $H(t)$ remains defined by Eq.~(\ref{dag}). It keeps the
quasi-Hermitian form,
 \be
  H^\dagger(t) \Theta(t)=
  \Theta^{}(t)H_{}(t)
 \,,
 \ \ \ \
 \Theta(t)=\Omega^\dagger(t)\Omega(t)\,,
 \label{dagobs}
 \ee
i.e., the status of an observable, viz., of an instantaneous energy.

\subsection{Schr\"{o}dinger equation for bra vectors\label{paramo}}

Using the same notation convention as above and abbreviating
$\Theta(t)|\psi(t)\kt \ \equiv \ |\psi_\Theta(t)\kt$ we replace the
ansatz of Eq.~(\ref{mapoge}) by its dual-space alternative
 \be
 |\psi^{}(t)\pkt
  =\left [\Omega^\dagger(t)\right ]^{-1}|\psi_\Theta(t)\kt
 \in {\cal H}^{(T)}\,,
 \ \ \ \ \ \ |\psi_\Theta(t)\kt \ \equiv \ \Theta(t)|\psi(t)\kt
  \in {\cal H}^{(F)}\,.
 \label{dumapo}
 \ee
This leads to the complementary Schr\"{o}dinger equation in ${\cal
H}^{(F)}$,
 \be
 {\rm i} \frac{\partial}{\partial t} \,|\psi_\Theta(t)\kt
 =G^\dagger(t)\,|\psi_\Theta(t)\kt\,
  \label{d2.3}
 \ee
(see \cite{timedep}; a misprint is to be removed in \cite{SIGMA}).

Whenever we decide to study the evolution of pure states, our pair
of Schr\"{o}dinger equations must be complemented by the choice of
initial values of $|\psi(t)\kt$ and $|\psi_\Theta(t)\kt$, say, at
$t=t_i=0$. This choice has the well known physical meaning
reflecting the experimental preparation of the quantum system in
question. In the non-Hermitian dyadic notation the pure states of
system ${\cal S}$ may and should be treated as represented by the
elementary projectors
 \be
 \pi_{\psi,\Theta}(t)=|\psi(t)\kt \,\frac{1}{\br
 \psi_\Theta(t)|\psi(t)\kt}\,\br \psi_\Theta(t)|\,.
 \ee
From here, one could also very quickly move to the non-Hermitian
statistical quantum mechanics where one prepares and works with the
statistical mixtures of states characterized, conveniently, by the
non-Hermitian density matrices of the form
 \be
 \widehat{\varrho}(t)=\sum_{k}|\psi^{(k)}(t)\kt
 \,\frac{p_k}{\br \psi^{(k)}_\Theta(t)|\psi^{(k)}(t)\kt}\,\br
 \psi^{(k)}_\Theta(t)|
 \,,
 \ \ \ \ \ \
 \sum_{k} p_k=1\,.
 \ee
For the time-independent preparation probabilities $p_k\neq p_k(t)$
one would then simply get, as an immediate consequence of
Eqs.~(\ref{SEFip}) and (\ref{d2.3}), the evolution equation
 \be
 {\rm
 i\,}\partial_t\, \widehat{\varrho}(t)= G(t)\widehat{\varrho}(t)
 -\widehat{\varrho}(t)\,G(t)\,,
 \ee
i.e., the non-Hermitian version of the Liouvillean evolution
picture.

\subsection{Heisenberg equations for
observables\label{koramo}}

Operator $Q_{}(t)$ of any observable must be compatible with the
requirement
 \be
 Q_{}^\dagger(t)\Theta(t)=\Theta(t)Q(t)
 \label{nenene}
 \ee
which guarantees its NIP observability inherited form the
Hermiticity of $\mathfrak{q}_{}(t) $. Formula (\ref{nenene})
represents just a transfer of the self-adjointness property $
\mathfrak{q}_{}(t) =\mathfrak{q}^{\dagger}(t)$ from ${\cal H}^{(T)}$
via the time-dependent generalization of the stationary definition
by Eq.~(\ref{dadag}),
 \be
 Q_{}(t)= \Omega_{}^{(-1)}(t)
 \mathfrak{q}_{(SP)}(t)\Omega_{}(t)\,.
 \label{redadag}
 \ee
The differentiation of Eq.~(\ref{redadag}) with respect to time
(denoted by overdot) now yields the Heisenberg-type evolution
equation
 \be
 {\rm i\,}\frac{\partial}{\partial t} \,{Q}_{}(t)=
  Q(t)\Sigma(t) -\Sigma_{}(t)Q(t)
 +K(t)\,,
 \ \ \ \ \ K(t)=\Omega_{}^{(-1)}(t)
 {\rm i\,}\dot{\mathfrak{q}}_{(SP)}(t)\Omega_{}(t)\,.
 \label{beda}
 \ee
We have two possibilities. In the general, non-stationary scenario
characterized by the nonvanishing time-derivative
$\dot{\mathfrak{q}}_{(SP)}(t)$ the evaluation of $K(t)$ (i.e., of
the necessary {\em input\,} information) would require the {\em
full\,} knowledge of the nonstationary Dyson maps $\Omega_{}(t)$.
Naturally, Heisenberg Eq.~(\ref{beda}) itself would be then
redundant because also the very operators $Q(t)$ would be
obtainable, by the same mapping, from their partners
${\mathfrak{q}}_{(SP)}(t)$ (which, by themselves, could be evaluated
from their derivatives, by ordinary integration).

In such a scenario one would not have any reason for leaving the
$T-$space. Let us, therefore, restrict our attention to the more
common, stationary models in which the partial derivatives
$\dot{\mathfrak{q}}_{(SP)}(t)$ vanish. This will make the recipe of
Eq.~(\ref{beda}) sensible because the related operator $K(t)$ would
be vanishing as well. Alternatively, we may also admit the models in
which a non-vanishing operator $K(t)$ would be prescribed in
advance. In both of these situations the solution of Heisenberg
Eq.~(\ref{beda}) would make sense.


\section{The physics of unitary evolution \label{secweb}}

\subsection{Broader context: Haag's theorem}

Schematically, the results of our preceding considerations may be
interpreted as certain preparatory steps towards a general
non-Hermitian interaction-picture recipe. In it one may relax all of
the {\it ad hoc} simplification assumptions, having to work just
with the three separate, independent unitary-evolution generators
$G(t)$ (for kets), $G^\dagger(t)$ (for their duals in ${\cal
H}^{(standard)}$) and $\Sigma(t)$ (Coriolis, for the operators of
observables).

\begin{figure}[h]                     
\begin{center}                         
\epsfig{file=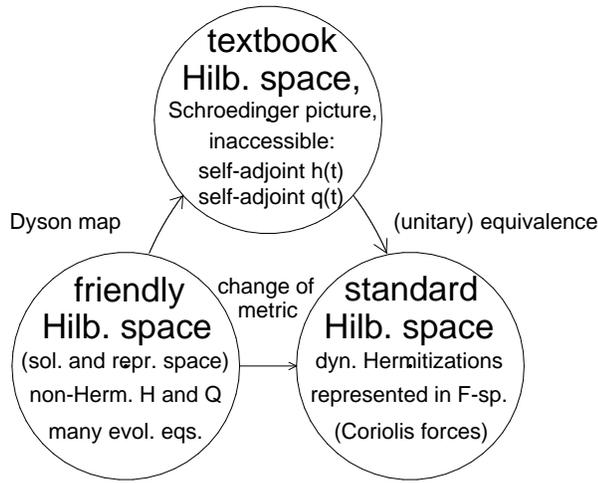,angle=270,width=0.7\textwidth}
\end{center}                         
\vspace{-2mm}\caption{The three Hilbert spaces and their respective
roles in the general non-stationary case.
 \label{picee3}}
\end{figure}

The situation reflecting the general non-stationary quantum dynamics
is summarized in Fig.~\ref{picee3}. We may find its conventional,
single-Hilbert-space quantum theory predecessor in the textbooks.
Indeed, once we choose $G(t)=H(t)$ (i.e., once we have $\Sigma(t)=0$
and the trivially evolving operators) we may speak about the
non-Hermitian version of Schr\"{o}dinger picture. Similarly, the
alternative special choice of $G(t)=0$ (yielding $\Sigma(t)=H(t)$
and the trivially evolving bras and kets) may be perceived as
determining the non-Hermitian version of the Heisenberg picture (cf.
Ref.~\cite{Heisenberg}).

Obviously, the general formalism finds its textbook analogue in the
so called interaction picture. A brief comment may prove useful in
this context. It is well known that the applicability of the
interaction picture is severely restricted,  in multiple quantum
field theories, by the famous Haag's theorem \cite{Haag}. The
theorem states, in essence, that due to the complicated mathematical
nature of these field theories, the above-mentioned one-to-one
correspondence between the two physical Hilbert spaces of
Fig.~\ref{picee3} may cease to exist, making the approach deeply
mathematically inconsistent. In our present language this would mean
that the Hilbert spaces ${\cal H}^{(textbook)}$ and ${\cal
H}^{(standard)}$ would happen to be unitarily inequivalent.
Fortunately, it is usually forgotten that the no-go nature of the
Haag's theorem does not extend to the ordinary quantum mechanics,
i.e., to the theory which is considered in our present paper.

\subsection{Initial values}

The solution of the triplet of the NIP evolution Eqs.~(\ref{SEFip}),
(\ref{d2.3}) and (\ref{beda}) can only lead to the meaningful
phenomenological predictions under the correct specification of the
initial values. Thus, the two kets $|\psi(t_i)\kt$ and
$|\psi_\Theta(t_i)\kt$ and the initial value of the observable of
interest (say, $Q(t_i)$) must be given in advance. The problem of
the specification of the initial operator value $Q(t_i)$ is subtle
because we must require that this operator also satisfies the
physical observability constraint~(\ref{nenene}) at the initial time
$t=t_i=0$.

From the point of view of mathematics the choice of the initial
values is arbitrary. In contrast, this choice carries the most
relevant information about physics. Thus, the non-Hermitian nature
of the NIP representation makes the guarantee of the consistency of
the initial values most important (cf., e.g., the numerical
experiments with the changes of initial values in \cite{Bila}).

\subsection{The measured quantities\label{seci3c}}

The most straightforward use of the NIP recipe would be based on the
assumption that we know both the generators $G(t)$ and $\Sigma(t)$.
The two Schr\"{o}dinger-type equations may be then solved to define
the two state vectors $|\psi(t)\kt$ and $|\psi_\Theta(t)\kt$ at all
times $t>t_i=0$. Similarly, the knowledge of $\Sigma(t)$ enables us
to recall Heisenberg Eq.~(\ref{beda}) and to reconstruct operators
$Q(t)$ up to the time of measurement $t=t_f>t_i=0$. In this manner
we can evaluate, in principle at least, the ultimate experimental
prediction
 \be
 \br \psi^{}_\Theta(t_f)|Q_{}(t_f) |\psi^{}(t_f) \kt\,.
 \label{redysMEAS}
  \ee
In contrast to the stationary case of Eq.~(\ref{dysMEAS}) we may
notice that the explicit knowledge of the {\em operator\,} of metric
becomes, within such a triple-equation recipe, redundant. Its role
is played by the {\em vector\,} $|\psi^{}_\Theta(t_f)\kt$. The
apparent paradox is easily explained via the trivial mathematical
identity
 \be
 {\rm i}\,\frac{\partial}{\partial t} \,
 \Theta(t)= \Theta(t)\Sigma(t) - \Sigma^\dagger(t) \Theta(t)\,.
 \label{identity}
 \ee
This shows that the knowledge of the metric is formally equivalent
to the knowledge of the Coriolis-force generator $\Sigma(t)$. It is
worth adding that due to Eq.~(\ref{dagobs}), the latter mathematical
identity is equivalent to its alternative version
 \be
 {\rm i}\,\frac{\partial}{\partial t} \,
 \Theta(t)= G^\dagger(t) \Theta(t)-\Theta(t)G(t) \,.
 \label{2.6}
 \ee
By several authors \cite{Wang,FringMou,Bila}, the latter
differential equation for operators was picked up as a tool of the
construction of the metric. Naturally, such a highly non-economical
recipe proved sufficiently efficient for the purposes of the study
of certain next-to-trivial illustrative non-Hermitian non-stationary
two-by-two matrix models in {\it loc. cit}.

\section{Special cases: Pre-selected generators $G(t)$
\label{seci3b}}

Quantum models using {\em any\,} form of the ``kinematical'' input
operator $G(t)$ will combine an enhanced flexibility with technical
complications. At least {\em some} of its applications may prove
{\em both\,} interesting and feasible. For example, in the light of
Theorem 2 of review paper \cite{ali} and of the related comments in
\cite{IJTP,FringMou}, the non-Hermitian interaction picture seems to
be {\em the only\,} formulation of quantum mechanics of unitary
systems in which the {\em non-Hermitian} versions of Schr\"{o}dinger
equations would  be allowed to contain  {\em time-dependent}, albeit
non-observable, Hamiltonian-like generators $G(t)$ and
$G^\dagger(t)$.

\subsection{Non-Hermitian Heisenberg picture ($G=0$)}

In Heisenberg picture (HP, cf.~Ref.~\cite{Heisenberg})  the
state-vectors are, by definition, stationary. Their generator is
trivial, $G_{(HP)}(t)\, \equiv \,0$. Due to Eq.~(\ref{2.6}), the HP
metric operator {\em cannot\,} be non-stationary, therefore. {\it
Vice versa,\,} {\bf any\,} description of dynamics during which the
metric would be changing requires a replacement of the
oversimplified HP formalism by its suitable (i.e., perhaps,
perturbatively tractable) NIP generalization.


In the non-Hermitian Heisenberg picture of Ref.~\cite{Heisenberg}
with vanishing $G=0$ both the bra and ket state vectors are
constant. This indicates that the isolated choice of any other
generator $G(t)$ has a more or less purely kinematical character and
need not carry {\em any\,} information about the dynamics.
Basically, the ``missing'' information about the dynamics will enter
the picture via the exhaustive description of the quantum system
${\cal S}$ at the initial instant $t=t_i$.

In the special non-Hermitian HP case the problem is slightly
simplified. From Eq.~(\ref{2.3}) we may deduce that
 $$\Sigma_{(HP)}(t)=H(t)$$
i.e., that the second generator becomes observable. Thus, we are
free to accept {\em any\,} form of such an observable Hamiltonian
$H(t)$ and/or {\em any\,} form of the other observable $Q(t)$ as an
independent, {\em additional\,} input information about the dynamics
of the quantum system in question.

In the non-Hermitian Heisenberg picture both of the NIP
Schr\"{o}dinger equations drop out and only the solution of the
Heisenberg Eqs.~(\ref{beda}) is asked for. Possibly, differential
Eq.~(\ref{defsig}) {defines} mapping $\Omega=\Omega(t)$ from its
pre-selected initial value at $t=t_i=0$. Such a constructive version
of the non-Hermitian Heisenberg picture already has its physical
applications \cite{cinani}. In the near future the non-Hermitian HP
and NIP formalisms might also prove needed and applied in quantum
cosmology \cite{ali,BBufab}.

\subsection{Extended non-Hermitian Heisenberg picture ($G \neq
G(t)$)}


The most straightforward generalization of the non-Hermitian HP
recipe may be based on the use of a {\em constant} nontrivial
generator of the evolution of wave functions. Such an ``extended
Heisenberg picture'' (EHP) with $G_{(EHP)}(t)=G_{(EHP)}(0)\neq 0$
has certain unexpected mathematical simplicity features (cf.~Eq.
Nr.~21 in~\cite{IJTP}). It proved also particularly suitable, in the
words of the Abstract of Ref.~\cite{IJTP}, for the description of
the ``evolution of the manifestly time-dependent self-adjoint
quantum Hamiltonians $\mathfrak{h}(t)$''.

Recently, both of these EHP applicability features were rediscovered
and illustrated, via Rabi-type Hamiltonian $\mathfrak{h}_{(SP)}(t)$,
in rapid communication~\cite{FrFrith}. The observation also renewed
the interest in the the systems (first discussed by B\'{\i}la
\cite{Bila}) for which the operator $G(t)$ is given in advance. In
the present setting, the key technical advantage of these quantum
systems is represented by the above-mentioned replacement of the
construction of the metric via Eqs.~(\ref{identity}) or (\ref{2.6})
by the mere solution of Schr\"{o}dinger Eq.~(\ref{d2.3}).

\section{General non-Hermitian
interaction picture\label{cwe}}

\subsection{Reconstruction of the basis}

One of the immediate consequences of the hypothetical knowledge of
the non-observable Hamiltonian $G(t)$ is that any change of the
initial ket-vector may be reflected by a new solution of
Schr\"{o}dinger Eq.~(\ref{SEFip}) and, {\it mutatis mutandis}, any
change of the initial bra-vector may be used as an initial value for
the parallel repeated solution of our second Schr\"{o}dinger
Eq.~(\ref{d2.3}). Once we decide to repeat the process $N-$times,
the knowledge of initial $N-$plets
 \be
  |\psi_1(0)\kt\,,|\psi_2(0)\kt\,,\ldots\,,|\psi_N(0)\kt\,
  \label{inijedna}
 \ee
and
 \be
  |\psi_{1,\Theta}(0)\kt\,,
  |\psi_{2,\Theta}(0)\kt\,,\ldots\,,|\psi_{N,\Theta}(0)\kt\,
  \label{inidva}
 \ee
will generate, in principle as well as in practice,
the related respective $N-$plets of state vectors
 \ben
  |\psi_1(t)\kt\,,|\psi_2(t)\kt\,,\ldots\,,|\psi_N(t)\kt\,
 \een
and
 \ben
  |\psi_{1,\Theta}(t)\kt\,,|\psi_{2,\Theta}(t)\kt\,,\ldots\,,
  |\psi_{N,\Theta}(t)\kt\,
 \een
at all times $t$.

One should now contemplate an arbitrary given physical quantum
system ${\cal S}$ which is prepared (i.e., known) at some initial
time (say, $t=0$) and which is to be analyzed in its NIP
representation, i.e., using a pre-selected operator $G(t)$. The
preparation of this system in {several (i.e., $N$) independent} pure
states is then equivalent to our knowledge of the respective initial
pairs of vectors $ |\psi_j(t)\kt$ and $|\psi_{j,\Theta}(t)\kt$. In
an extreme case we may assume that the initial $N-$plet  (with $N
\leq \infty$ in general) is, at $t=0$, bi-orthonormal,
 \be
  \br \psi_{m,\Theta}(0) |\psi_{n}(0) \kt = \delta_{m,n}\,,\
  \ m,n = 1, 2, \ldots, N\,
 \label{complean0}
 \ee
and {complete},
 \be
 \sum_{n=1}^N |\psi_{n}(0)\kt \br \psi_{n,\Theta}(0) | = I\,.
 \label{comple0}
 \ee
The assumption of our knowledge of $G(t)$ at {\em all} $t>0$ then
enables us to construct the pairs $ |\psi_j(t)\kt$ and
$|\psi_{j,\Theta}(t)\kt$ such that
 \be
 \sum_{n=1}^N |\psi_{n}(t)\kt \br \psi_{n,\Theta}(t) | = I\,,
 \ \ \ \ \ \
 \br \psi_{1,\Theta}(t) |\psi_{n}(t) \kt = \delta_{m,n}\,,\
  \ m,n = 1, 2, \ldots, N\,.
 \label{comple}
 \ee
Their respective completeness and bi-orthonormality {\em survive} at
all times.

\subsection{Reconstruction of the metric\label{emxpiy}}

The hypothetical NIP representability of a given quantum system
${\cal S}$ must reflect the existence of its (by assumption,
overcomplicated and technically inaccessible but existing)
Hermitian, textbook representation using the self-adjoint operators
of observables in ${\cal H}^{(T)}$. The system ${\cal S}$ must have
its non-Hermitian SP representation at $t=0$. The $N-$plets
(\ref{inijedna}) and (\ref{inidva}) of pure states may be expected
to be prepared, at $t=0$, as the respective left and right
eigenstates of a relevant observable. Typically, we work with a
Hamiltonian and obtain its spectral representation, therefore,
 \be
 H(t) = \sum_{n=1}^N\,|\psi_{n}(t)\kt E_n(t) \br \psi_{n,\Theta}(t)
 |\,.
 \label{lefdi}
 \ee
Such an instantaneous energy operator is, by definition, related to
its isospectral Hamiltonian partner $\mathfrak{h}$ which is
self-adjoint in ${\cal H}^{(T)}$. The same parallelism will also
apply to any other observable, i.e., in our notation, to any
operator $Q(t)$.

In the light of paragraph \ref{koramo} we will consider the
``tractable'' scenarios in which the SP energies themselves are
conserved, i.e., $ E_n(t)= E_n(0)=E_n$. The key benefit of this
assumption is that we may now recall Ref.~\cite{SIGMAdva} and get an
important mathematical result, viz., the solution of the
time-dependent quasi-Hermiticity relation (\ref{2.6}) at all times.

\begin{thm}
 \label{lemoine}
For a given generator $G(t)$ and for the two initial vector sets
(\ref{inijedna}) and (\ref{inidva}) with properties
(\ref{complean0}) and (\ref{comple0}), the metric operator
$\Theta(t)$ has the following formal representation in ${\cal
H}^{(F)}$,
 \be
 \Theta(t)=\sum_{n=1}^N
 |\psi_{n,\Theta}(t)\kt \,
  \br \psi_{n,\Theta}(t)
 | \,.
 \label{defifi}
 \ee
\end{thm}
 \begin{proof}
Along the lines outlined in Ref.~\cite{SIGMAdva}, formula
(\ref{defifi}) follows from completeness (\ref{comple}) and from the
definition of the physical bra vectors $ \br \psi_{n,\Theta}(t)| $
in ${\cal H}^{(S)}$ (the superscript abbreviates the ``standard''
physical Hilbert space, unitarily equivalent to ${\cal H}^{(T)}$).
The representation is formal because in some pathological cases with
$N = \infty$ its right-hand-side series need not converge. A less
formal extension of validity of the theory to these cases would
require a more rigorous mathematical specification of the properties
of the biorthonormal bases \cite{fabio}.
  \end{proof}

We see that via the repeated solution of Schr\"{o}dinger equations
one can obtain the spectral-like formula for the metric ``for
free''. Formula (\ref{defifi}) indicates that the initial choice of
the biorthonormal basis determines the unique initial-value choice
of the metric $\Theta(t_i)$ with $t_i=0$ as its byproduct. {\it Vice
versa}, by definition, the initial choice of $\Theta(0)$ would
enable us to generate the set of initial bras (\ref{inidva}) from a
pre-selected set (\ref{inijedna}) of initial kets. The condition of
the mutual compatibility (\ref{dagobs}) of $H(0)$ with $\Theta(0)$
of Eqs.~(\ref{lefdi}) and (\ref{defifi}) follows from properties
(\ref{comple}). The important remaining freedom is that the energies
$E_n$ in the standard spectral representation (\ref{lefdi}) of the
observable Hamiltonian are free parameters of the theory. Their
numerical values remain unrestricted in the NIP theoretical
framework, therefore.

\subsection{Reconstruction of the Dyson mappings $\Omega(t)$}

The operator of metric (\ref{defifi}) must be self-adjoint and
strictly positive in ${\cal H}^{(F)}$ so that it can be
diagonalized, under suitable mathematical conditions, via a unitary
operator ${\cal U}(t)$,
 \be
 \Theta(t)={\cal U}^\dagger(t)\theta^2(t) {\cal U}(t)\,.
 \ee
The matrix $\theta^2(t)$ of the real and positive eigenvalues of
$\Theta(t)$ and its square root $\theta(t)$ are both diagonal, real,
finite and non-vanishing. The diagonalization of the metric is
purely numerical. In practice we either work with finite-dimensional
Hilbert spaces, or we truncate them to a finite dimension
$N<\infty$.

Up to another, independent unitary-matrix ambiguity ${\cal V}(t)$,
we may factorize the metric into the product of the entirely general
class of non-stationary Dyson's mappings,
 \be
  \Omega(t)
 ={\cal V}^\dagger(t)\theta(t) {\cal U}(t)
 \,.
 \label{ome}
 \ee
At this moment we are already able to construct the Coriolis force
$\Sigma(t)$ so that we may construct a new, tilded-operator sum
 $
 \widetilde{H}(t)=G(t)+\Sigma(t)
 $.
It is instructive to verify that the latter operator is equal to the
old untilded unobservable Hamiltonian of Eq.~(\ref{lefdi}). One
reveals that the demonstration of the identity
$\widetilde{H}(t)={H}(t)$ is easy due to the validity of
Theorem~\ref{lemoine}.

\section{Application in relativistic quantum mechanics
\label{seci3}}

The authors of the textbooks on quantum theory are very well aware
that a consequent incorporation of all of the effects of the
relativistic kinematics requires a transition from quantum mechanics
to the quantum field theories in which the number of particles is
not conserved. The relativistic versions of the models in quantum
mechanics are considered as certain sufficiently satisfactory
approximations of the physical reality. The choice of the Dirac
equation describing fermions is preferred. The reasons range from
the applicability of Pauli principle (explaining antifermions as
holes in the Dirac's sea) up to an agreement of the theory with
experiments, say, for hydrogen atom \cite{Fluegge}.

\subsection{Stationary Klein-Gordon equation}


Within the framework of the relativistic quantum mechanics the
widespread preference of the study of the Dirac equation is,
certainly, a paradox because its bosonic Klein-Gordon (KG)
alternative is less mathematically complicated. In one of its
simplest versions written in units $\hbar=c=1$,
 \be
 \left (
 \frac{\partial^2}{\partial t^2}+D\right
 )\,\psi^{(KG)}(\vec{x},t)=0\,,
 \ \ \ \ \ \ D=-\triangle +m^2
  \label{kgha}
 \ee
one does not consider any external electromagnetic field so that the
kinetic energy is represented by the elementary Laplacean
$\triangle$. The dynamics is also reduced to a scalar external field
simulated by a suitable position-dependence of the mass term,
$m^2=m^2(\vec{x})$.

The most common physical interpretation of Eq.~(\ref{kgha}) was
proposed by Feshbach and Villars (FV, \cite{FV}). They changed the
variables
 \be
 \psi^{(KG)}(\vec{x},t) \ \ \to \ \
 \br \vec{x}|\psi^{(FV)}(t)\kt=
 \left (
 \ba
 {\rm i}\partial_t \psi^{(KG)}(\vec{x},t)\\
 \psi^{(KG)}(\vec{x},t)
 \ea
 \right )\,
  \label{assi}
 \ee
(cf. also Ref.~\cite{KGja}) and they replaced the hyperbolic partial
differential Eq.~(\ref{kgha}) by its parabolic partial differential
mathematical equivalent
 \be
 {\rm i} \frac{\partial}{\partial t} \,
 \,|\psi^{(FV)}(t)\kt=H_{(FV)}\,|\psi^{(FV)}(t)\kt
  \,.
 \ee
This is an evolution equation with stationary generator
 \be
 \ \ \
H_{(FV)}=
 \left (
 \begin{array}{cc}
 0&D\\
 I&0
 \ea
 \right )\neq H_{(FV)}(t)\,.
  \label{SEFekg}
 \ee
In the most common Hilbert space of states ${\cal H}^{(FV)}= {\cal
L}^2(\mathbb{R}^3)\bigoplus{\cal   L}^2(\mathbb{R}^3)$ such a
generator of the evolution of the FV wave functions cannot be
interpreted as standard Hamiltonian because it is manifestly
non-Hermitian, $ H_{(FV)} \neq H_{(FV)}^\dagger
 $.

A partial resolution of the puzzle was suggested by Pauli and
Weisskopf \cite{PW}. They noticed that, translated to the modern
language, the non-Hermitian operator $H_{(FV)}$ is ${\cal
PT}-$symmetric {\it alias\,} self-adjoint in an {\it ad hoc\,} Krein
space. This means that this operator is tractable as self-adjoint
with respect to a new, auxiliary, indefinite inner product,
 \be
 \br \psi_1|\psi_2\kt \ \ \to \ \
 \left ( \psi_1,\psi_2 \right )_{(Krein)}
 =\br \psi_1|{\cal P}_{(FV)}|\psi_2 \kt
\,.
 \label{tran}
 \ee
On this background, after the standard introduction of an external
electromagnetic field, the model can be perceived as a
charge-conserving evolution of the physical KG field in which the
number of particles is not conserved (for details see, e.g., chapter
XII of \cite{Const}).

The alternative, more natural Mostafazadeh's proposal (cf.
\cite{[150],[151]}) was based on the stationarity property
(\ref{SEFekg}). The idea consisted in the further change of the
inner product,
 \be
 \left ( \psi_1,\psi_2 \right )_{(Krein)}\ \ \to \ \
 \left ( \psi_1,\psi_2 \right )_{(Mostafazadeh)}
 =\br \psi_1|\Theta_{(stationary)}|\psi_2 \kt\,.
 \label{trans}
 \ee
The physics-restoring operators
$\Theta_{(stationary)}=\Theta_{(stationary)}^\dagger$ were required
{\em positive definite} (cf. also \cite{Geyer} and \cite{ali} in
this context). The initial, manifestly unphysical Hilbert space
${\cal H}^{(FV)}$ (with inner products $\br \psi_1|\psi_2\kt$) was
replaced by the Mostafazadeh's Hilbert space ${\cal H}^{(M)}$ which
differed from ${\cal H}^{(FV)}$ by the correct and physical inner
product (\ref{trans}) defined in terms of the so called metric
operator $\Theta_{(stationary)}$ \cite{Dieudonne}.
Eq.~(\ref{SEFekg}) was reinterpreted as Hamiltonian living in
physical Hilbert space ${\cal H}^{(M)}$ with inner product
(\ref{trans}).

In opposite direction, {\em any\,} Hamiltonian (\ref{SEFekg}) with
property
 \be
  H_{(FV)}^\dagger \Theta_{(stationary)}=\Theta_{(stationary)}
  H_{(FV)}\,
  \label{relac}
  \ee
may be perceived as self-adjoint with respect to the {\it ad hoc}
product (\ref{trans}). In the Klein-Gordon case, in particular,
relation (\ref{relac}) may be satisfied by the Hilbert-space metrics
of the closed block-matrix form
 \be
 \Theta_{(stationary)}=
 \left (
 \begin{array}{cc}
 1/\sqrt{D}&0\\
 0&\sqrt{D}
 \ea
 \right )\,.
 \label{thassi}
 \ee
Formulae (\ref{trans}) and (\ref{thassi}) yield the first-quantized
interpretation of the Klein-Gordon Eq.~(\ref{kgha}) in which the
probability density of finding the relativistic spinless massive
particle is never negative.

\subsection{Non-stationary scenario}

From the point of view of relativistic kinematics the stationary
position-dependence $m^2=m^2(\vec{x})$ of the mass term violates the
Lorentz covariance of the interaction. A privileged coordinate frame
must be used so that a weakening $m^2 = m^2(\vec{x},t)$ of the
assumption would be highly desirable.

\subsubsection{The NIP generator $G(t)$ of KG ket-vectors}

The acceptance of the hypothesis $m^2 = m^2(\vec{x},t)$ leads to the
replacement of Eq.~(\ref{kgha}) by its non-stationary form
 \be
 \left (
 \frac{\partial^2}{\partial t^2}+D(t)\right
 )\,\psi^{(KG)}(\vec{x},t)=0\,,
 \ \ \ \ \ \ D(t)=-\triangle +m^2(\vec{x},t)\,.
  \label{kghage}
 \ee
The same change of variables as above
defines the NIP ket-vector wave functions
 \be
  \br \vec{x}|\psi^{(NIP)}(t)\kt=
 \left (
 \ba
 {\rm i}\partial_t \psi^{(KG)}(\vec{x},t)\\
 \psi^{(KG)}(\vec{x},t)
 \ea
 \right )\,.
  \label{prassi}
 \ee
Its use converts Eq.~(\ref{kghage}) into the KG realization of the
non-stationary Schr\"{o}dinger Eq.~(\ref{SEFip}),
 \be
 {\rm i} \frac{\partial}{\partial t} \,|\psi^{(NIP)}(t)\kt=
\left (
 \begin{array}{cc}
 0&D(t)\\
 I&0
 \ea
 \right )\,|\psi^{(NIP)}(t)\kt\,
  \label{niSEFip}
 \ee
This equation contains the KG generator $ G_{(NIP)}(t)$.  Its
spectrum may, but need not, be real \cite{SIGMA}. In the current
literature it is called ``generator'' \cite{timedep,SIGMA},
``Hamiltonian'' \cite{Wang,Bila}, ``nonobservable Hamiltonian''
\cite{FringMou} or ``unobservable Hamiltonian'' \cite{FrFrith}.

\subsubsection{The NIP generator $G^\dagger(t)$ of KG bra-vectors}

B\'{\i}la's preprint \cite{Bila} deserves to be recalled as one of
the first studies of the methodical scenario in which one starts the
analysis of quantum dynamics from the knowledge of $G(t)$. Using an
elementary toy-model generator $G(t)$ the author solved
Eq.~(\ref{2.6}) and revealed that the qualitative features of the
metric $\Theta(t)$ depend rather strongly on the choice of this
metric operator at the initial time $t=t_i$. Later, similar
constructive case studies were published in
Refs.~\cite{IJTP,Wang,FringMou}.

The weak point of the B\'{\i}la's preprint lies in the complicated
operator form of equation~(\ref{2.6}). This key to evolution was
later called ``time-dependent quasi-Hermiticity relation''
\cite{FringMou}. In paragraph \ref{paramo} above we recommended a
more economical approach. We felt inspired by our recent studies
\cite{IJTP,Heisenberg} in which we analyzed the consequences of the
choice of a vanishing or constant $G(t)$, respectively. The role of
the initial conditions proved more important than expected while the
role of the generator itself appeared to be merely technical.
Moreover, the use of the non-stationary metric $\Theta(t)$ was
replaced by the solutions of the second Schr\"{o}dinger equation.

We came to the conclusion that one may skip the solution of
Eq.~(\ref{2.6}). In the implementation of the NIP recipe (based on
the knowledge of metric) one is strongly recommended to replace the
multiplicative definition of
 \be
 |\psi_\Theta^{(NIP)}\kt = \Theta(t)\,|\psi^{(NIP)}\kt
 \ee
by the KG form of the second Schr\"{o}dinger
Eq.~(\ref{d2.3}),
 \be
 {\rm i} \frac{\partial}{\partial t} \,|\psi_\Theta^{(NIP)}(t)\kt=
\left (
 \begin{array}{cc}
 0&I\\
 D^*(t)&0
 \ea
 \right )\,|\psi_\Theta^{(NIP)}(t)\kt\,.
  \label{niSEFip}
 \ee
It is worth noticing that one could even admit here the complex
effective mass terms $m^2(\vec{x},t) \notin \mathbb{R}$.

\subsubsection{The Coriolis NIP generator $\Sigma(t)$}

The replacement of the construction of operators by the construction
of vectors will simplify the calculations. This will also render
them feasible far beyond the most popular two-by-two matrix toy
models. One can predict that in the nearest future such a form of
the NIP formalism might also find multiple other applications, most
of which could closely parallel the existing uses of Hermitian
interaction picture. One can expect that in a way guided by the
parallel an enhanced attention will be paid to the models with the
dynamics dominated by the HP-resembling non-Hermitian generator
$\Sigma(t)$. Naturally, after one pre-selects $G(t)$, the choice of
$\Sigma(t)$ ceases to be unconstrained, mainly due to the deeply
physical role played by the experiment-related preparation of the
initial-state vectors (\ref{inijedna}) and (\ref{inidva}) and of the
related initial-instant energy $H^{(NIP)}(t_i)$ as prescribed by
Eq.~(\ref{lefdi}). Moreover, we know that the evolution of the
latter observable proceeds via the Heisenberg-type Eq.~(\ref{beda}),
i.e.,
 \be
 {\rm i\,}\frac{\partial}{\partial t} \,{H}^{(NIP)}(t)=
  H^{(NIP)}(t)\Sigma^{(NIP)}(t)
  -\Sigma^{(NIP)}(t)H^{(NIP)}(t)+K^{(NIP)}(t)
 \label{bedakat}
 \ee
or, equivalently,
 \be
 {\rm i\,}\frac{\partial}{\partial t} \,{H}^{(NIP)}(t)=
     G^{(NIP)}(t) H^{(NIP)}(t)
     -H^{(NIP)}(t)G^{(NIP)}(t)+K^{(NIP)}(t)
 \label{rebedakat}
 \ee
with some vanishing or pre-determined SP-variability term
 $$
K^{(NIP)}(t)=\Omega_{}^{(-1)}(t) {\rm
i\,}\dot{\mathfrak{h}}_{(SP)}(t)\Omega_{}(t)\,.$$ As long as
$\Sigma^{(NIP)}(t) = H^{(NIP)}(t) - G^{(NIP)}(t)$ there remains no
freedom in the choice of this operator. Subsequently, as long as we
have, from definition,
 \be
 {\rm i} \frac{\partial}{\partial t} \,\Omega^{(NIP)}(t)\kt=
\Omega^{(NIP)}(t)\,\Sigma^{(NIP)}(t)\,,
  \label{reniSEFip}
 \ee
the only ambiguity of $\Omega^{(NIP)}(t)$ is contained in its
initial-value specification.

\section{Discussion  \label{seci4}}

\subsection{Open problems\label{adiaba}}

In the non-stationary NIP formalism our intuition need not work and
our expectations may prove wrong. For this reason it is extremely
fortunate that one can simulate many unusual features of quantum
systems ${\cal S}$  using classical optics \cite{Samsonov2}.
Moreover, the results of the experimental tests and classical
optical simulations may also find an alternative, manifestly
non-unitary-evolution interpretations in the effective-operator
descriptions of open quantum systems \cite{Samsonov}.

In all of these dynamical regimes one of the most interesting open
questions is the problem of the domain of the survival of validity
of the usual adiabatic hypothesis. This hypothesis may be violated,
and violated in an unexpected manner~\cite{Stefan}. Thus, even in
the domain of the safely unitary quantum evolutions our
understanding of the limits of the validity of the adiabatic
approximation seems to be far from complete. The obstacles
encountered in such a context were summarized in preprint
\cite{Bila}. They may be separated into three subcategories. In the
first one one deals with the problems of {\em incompleteness} of the
information about dynamics \cite{Jones}. Even if we start from a
trivial $G(t)$, the whole NIP formalism degenerates to the
non-Hermitian Heisenberg picture~\cite{Heisenberg,cinani}. This
implies that the answers to all of the questions about stability
and/or instability remain strongly dependent on the number of
additional assumptions about the dynamics.

In the second subcategory of problems one encounters an {\em
ambiguity\,} related to the choice of the representation of the
correct physical Hilbert space ${\cal H}^{(S)}$. This form of
flexibility already represented a serious methodical challenge in
the traditional stationary models \cite{Geyer}). In the present
non-stationary context it reflects the freedom of preparation of the
system ${\cal S}$ in a pure state at $t=0$. It was again  B\'{\i}la
\cite{Bila} who addressed this problem in 2009. He illustrated the
relevance of the initial conditions via a two-by-two-matrix toy
model. The resulting demonstration of the existence of deeply
different evolution scenarios was sampled in Fig. Nr. 1 of {\it loc.
cit.}. Recently, a similar approach and analogous analysis of
another two-by-two matrix model were published in
Ref.~\cite{FringMou}.

The third group of the currently open questions concerns the {\em
economy} of the alternative construction strategies. In this sense
we offered here a new approach, circumventing partially the main
weakness of the above-cited constructions which were all based on
the solution of the operator evolution Eq.~(\ref{2.6}). Even in its
truncated, $N-$dimensional Hilbert space approximation this equation
is a rather complicated coupled set of ordinary differential
equations for as many as $N^2$ unknown matrix elements of
$\Theta(t)$. Although we managed to weaken the difficulty, we had to
postpone the practical numerical tests of the economy of the present
approach to a separate study. After all, the difficulty of such a
project is underlined by the freshmost study \cite{FringMou} in
which the solution of the operator evolution Eq.~(\ref{2.6})
(without a help by an ansatz for $\Omega(t)$) was only performed at
$N=2$.

\begin{table}[h]
\caption{Sample of notation conventions used in the literature. }
\label{pexp4}

\vspace{2mm}

\centering
\begin{tabular}{||c|l||c|c|c||}
\hline \hline
    {\rm symbol}&
    {\rm meaning} &
    {\rm Ref. \cite{ali}} &
    {\rm Ref. \cite{FringMou}} &
    {\rm Ref. \cite{Geyer}}
    \\
    \hline \hline
 $\Omega$ &
  {\rm Dyson's map, see Eq.~(\ref{mapoge})}&  $\rho$  & $\eta$&$S$ \\
 $\Theta$ & $=\Omega^\dagger\Omega$, the
   {\rm metric} &  $\eta_+$  & $\rho$&$\widetilde{T}$\\
 $|\psi\kt$
 &  {\rm state vector, Eq.~(\ref{SEFip})}&  $|\psi\kt$  & $\Psi$ &$|\Psi\kt$\\
 $G$&  {\rm the generator of kets, Eq.~(\ref{2.3})} &  $-$  & $H$&$-$ \\
 $\mathfrak{h}$&
  {\rm textbook Hamiltonian, Eq.~(\ref{udagobs})}&  $h$  & $h$ &$-$ \\
 $\Sigma$&  {\rm Coriolis Hamiltonian, Eq.~(\ref{defsig})}&  $-$  & unabbr.  &$-$\\
 $H$&
  {\rm observable Hamiltonian, Eq.~(\ref{dagobs})} &  $H$  & $\widetilde{H}$&$H$ \\
 $|\psi_\Theta\kt$&  {\rm dual state vector, Eq.~(\ref{dumapo})} &  $|\phi\kt$  &
    $\rho\Psi$ &$\widetilde{T}|\Psi\kt$\\
 \hline
 \hline
\end{tabular}
\end{table}

\subsection{Terminology: A Rosetta stone}

One of the consequences of the novelty of the NIP picture in both of
its stationary and non-stationary versions is the diversity of
notation (see Table~\ref{pexp4}). This is partly due to the
replacement of the traditional Hamiltonian in ${\cal H}^{(T)}$ by
the triplet of operators acting in ${\cal H}^{(F)}$. All of them
participate in the description of the unitary evolution of a given
system ${\cal S}$ so that all of them could be called
``Hamiltonians''. In \cite{SIGMA} our present convention was
developed to distinguish, without subscripts or superscripts,
between the observable $H(t)$ of Eq.~(\ref{udagobs}) and the two not
necessarily observable generators $G(t)$ (changing the state-vector
kets via Schr\"{o}dinger-type Eq.~(\ref{2.3})) and $\Sigma(t)$
(prescribing the time-dependence of every operator of an observable
via the respective Heisenberg-type Eq.~(\ref{beda})).

In the non-stationary NIP formalism it is rather counterintuitive
that the generators $G(t)$ and $\Sigma(t)$ are, in general, not
observable while the observable Hamiltonian $H(t)$ itself is just
their sum. Some authors decided to treat $G(t)$ as ``the''
Hamiltonian \cite{FringMou,Bila,FrFrith}. No comments are usually
made on the fact that its spectrum is, in general,  complex even if
the evolution of system ${\cal S}$ itself is unitary.


Besides the single-observable scope of predictions (\ref{redysMEAS})
we are often interested in the knowledge of metric for some other
reasons. A few of them may be found listed in Ref.~\cite{Batal}. In
paragraph \ref{neramo} we also emphasized the formal equivalence
between the respective ``time-dependent'' and ``time-independent''
quasi-Hermiticity relations (\ref{2.6}) and (\ref{dagobs}) for the
metric. Now we have to add that this is certainly not an equivalence
from the point of view of applications. The latter relation is
algebraic while the former one is differential, much more
complicated to solve. For this reason, even when one needs to know
the metric, its reconstruction should proceed via
Eq.~(\ref{dagobs}). The preference of Eq.~(\ref{2.6}) by some of the
above-cited authors looks rather naive.

A slightly discouraging aspect of Eq.~(\ref{dagobs}) may be seen in
the fact that one has to determine {\em both} of the ``unknown''
operator components $H(t)$ and $\Theta(t)$ of Eq.~(\ref{dagobs}) at
once. A hint to the resolution of the apparent paradox was provided
in paragraph \ref{koramo} above. We argued there that the observable
Hamiltonian $H(t)$ must be perceived as carrying a {\em necessary}
addition to the input information about dynamics. In other words,
even after an exhaustive exploitation of the knowledge of $G(t)$ the
definition of quantum system ${\cal S}$ remains as incomplete as in
the above-mentioned Heisenberg picture where $G_{(HP)}(t)=0$.

\section{Summary  \label{seci5}}

Any quantum system ${\cal S}$ can be described using any formulation
{\it alias} picture, in principle at least \cite{Messiah}. The
textbooks offer a more or less standard menu of pictures, each one
of which may prove technically most appropriate for a specific
${\cal S}$ \cite{nine}. {\it Vice versa}, the development of new
pictures of unitary evolution may help to extend the class of
tractable quantum systems ${\cal S}$. Our paper may be read as an
illustration of the latter statement. We recalled the
``three-Hilbert-space'' picture of Ref.~\cite{SIGMA} and we
demonstrated that its NIP reformulation may offer new insights in
relativistic quantum mechanics.

Our specific task of a consistent and constructive interpretation of
non-stationary Klein-Gordon quantum-mechanics systems had several
aspects. The most important one may be seen in having us forced to
work with Schr\"{o}dinger equations containing unobservable
generators $G(t)$ and/or $\Sigma(t)$. This is a methodical feature
of the new theory which looks counterintuitive \cite{web}. After one
admits that the Dyson's maps can be time-dependent, the acceptance
of the unobservability becomes natural, understandable and well
founded \cite{SIGMA,FringMou,which}.

A few equally important phenomenological challenges connected with
our present first-quantization approach to the non-stationary
Klein-Gordon system were listed as open problems. We provided
several answers to some of them. Incidentally, we improved also the
computational {\em economy\,} of the evaluation of the experimental
predictions (\ref{redysMEAS}). The feasibility of their computation
may be perceived as a very core of the theory. Typically, the
authors of Refs.~\cite{FringMou,Bila,FrFrith} used these prediction
formulae in combination with the metric-multiplication definition
(\ref{amelio}) where one needs to know the metric. Thus, these
authors had to solve the {\em operator\,} evolution
Eq.~(\ref{2.6})). In contrast we noticed that the relevant
information about the metric is all carried by the {\em single
bra-vector} $\br \psi_\Theta(t)|$. Thus, we replaced the latter
recipe by the direct solution of the {\em vector\,}-evolution
Schr\"{o}dinger-type Eq.~(\ref{d2.3}). In this way we eliminated the
complicated, Heisenberg-equation-resembling evolution
rule~(\ref{2.6}) as redundant. We explained that such an operator
differential-equation rule is just a combination of identity
(\ref{identity}) with the much simpler, more natural and fully
transparent algebraic condition of observability (\ref{dagobs}).

We may summarize that the key mathematical feature of our present
recipe is that it prescribes a specific time-evolution of the
physical inner product in the Hilbert space of states. This leads to
the standard probabilistic interpretation of the theory while
opening new perspectives. In the context of physics, one of the most
concise characteristics of our present approach to the
non-stationary Klein-Gordon problem may be seen in its parallelism
to the conventional interaction picture. The main difference from
the latter formulation of quantum dynamics (and from its
non-Hermitian stationary SP formulations \cite{ali}) is that we
admit that the generators $G(t)$ and $\Sigma(t)$ need not be
observable in general. Still, the unitarity of the underlying
physical system is preserved, guaranteed by the fact that the
manifestly non-unitary evolution of $|\psi^{(NIP)}(t) \rangle$
(controlled by a Schr\"{o}dinger-type equation) is {\em strictly
compensated} by the non-unitarity of the evolution of the operators
of observables (controlled, in parallel, by non-Hermitian
Heisenberg-type equations).

\section*{Acknowledgement}

Research supported by the GA\v{C}R Grant Nr. 16-22945S.


\newpage

\begin{thebibliography}{99}


\bibitem{Messiah}
A. Messiah, Quantum Mechanics I. Amsterdam, North Holland, 1961.

\bibitem{Haag}
R. Haag,
Matematisk-fysiske Meddelelser
 29, 12 (1955);

see also
 https://en.wikipedia.org/wiki/Haag\%27s\_theorem

\bibitem{Hall}
D. Hall and A. S. Wightman,
 Matematisk-fysiske Meddelelser 31, 1 (1957);

M. Reed and B. Simon,
 Methods of modern mathematical physics II.
 Academic Press,
 New York, 1975, Theorem X.46.


\bibitem{nine}
D. F. Styer et al,
Am. J. Phys. 70, 288 (2002).


\bibitem{Dyson}
F. J. Dyson, Phys Rev 102, 1217 (1956).


\bibitem{Geyer}
F. G. Scholtz, H. B. Geyer and F. J. W. Hahne, Ann. Phys. (NY) 213,
74 (1992).

\bibitem{Carl}
C. M. Bender,
Rep. Prog. Phys. 70, 947 (2007).


\bibitem{ali}
A. Mostafazadeh,
Int. J. Geom. Meth. Mod. Phys. 7, 1191
 (2010).

\bibitem{book}
%
F. Bagarello, J.-P. Gazeau, F. H. Szafraniec, and M. Znojil,
Non-Selfadjoint Operators in Quantum Physics: Mathematical Aspects.
Wiley, Hoboken, 2015.



\bibitem{MZbook}
M. Znojil, in
Ref.~\cite{book}, pp. 7 - 58.

%



\bibitem{Fringlas}
C. Figueira de Morisson Faria and A. Fring, J. Phys. A: Math. Gen.
39, 9269 (2006);

C. Figueira de Morisson Faria and A. Fring,
Laser Phys. 17, 424 (2007).


\bibitem{PLB}
A. Mostafazadeh, Phys. Lett. B 650, 208 (2007).


\bibitem{web}
A. Mostafazadeh,
arXiv:
quant-phys/0711.0137;

A. Mostafazadeh,
arXiv: quant-phys/0711.1078.

\bibitem{timedep}
M. Znojil, %
%
Phys. Rev. D 78, 085003 (2008).
%


\bibitem{SIGMA}
M. Znojil, SIGMA 5, 001 (2009) (e-print overlay: arXiv:0901.0700).

\bibitem{IJTP}
M. Znojil,
Int. J. Theor. Phys. 52, 2038 (2013).

\bibitem{Luiz}
%
 F. S. Luiz, M. A. Pontes and M. H. Y. Moussa,
arXiv: 1611.08286.

\bibitem{Heisenberg}
M. Znojil, 
Phys. Lett. A 379, 2013 (2015).

\bibitem{symmetry}
M. Znojil,
Symmetry 8, 52 (2016).


\bibitem{cinani}
Y.-G. Miao and Zh.-M. Xu, Phys. Lett. A 380, 1805 (2016).


\bibitem{Wang}
J.-B. Gong and Q.-H. Wang, Phys. Rev. A 82, 012103 (2010);

J.-B. Gong and Q.-H. Wang,
J. Phys. A: Math. Theor. 46, 485302 (2013).



\bibitem{Maamache}
%
M. Maamache, Phys. Rev. A 92,  032106 (2015);
%

B. Khantoul, A. Bounames and M. Maamache,
 arXiv: 1610.09273.
%


\bibitem{FringMou}
A. Fring and M. H. Y. Moussa,
%
Phys. Rev. A 93, 042114 (2016);

A. Fring and M. H. Y. Moussa,
%
Phys. Rev. A 94, 042128 (2016).
%




\bibitem{Stone}
M. H. Stone,
Ann. Math. 33,  643 (1932).



\bibitem{Batal}
A. Mostafazadeh and A. Batal, J. Phys. A: Math. Gen. 37, 11645
(2004).



\bibitem{BB}
C. M. Bender and S. Boettcher, Phys. Rev. Lett. 80, 5243 (1998).



\bibitem{arabky}
M. Znojil, I. Semoradova, F. Ruzicka, H. Moulla and I. Leghrib,
Phys. Rev. A 95, 042122 (2017).



\bibitem{which}
M. Znojil, 
arXiv:
quant-ph 0711.0535.



\bibitem{Bila}
 H. B\'{\i}la,
 Non-Hermitian Operators in
Quantum Physics. Charles University, Prague, 2008 (PhD thesis);

H. B\'{\i}la,
 arXiv: 0902.0474.

\bibitem{Stefan}
T. J. Milburn, J. Doppler, C. A. Holmes, S. Portolan, S. Rotter and
P. Rabl,
%
    Phys. Rev. A 92, 052124 (2015).




\bibitem{BBufab}
M. Znojil, J. Phys. Conf. Ser. 343, 012136 (2012);



M. Znojil, "Quantization of Big Bang in crypto-Hermitian Heisenberg
picture",  in F. Bagarello, R. Passante and C. Trapani, eds,
"Non-Hermitian Hamiltonians in Quantum Physics". Springer, Cham,
2016,
pp. 383 - 399.



\bibitem{FrFrith}
A. Fring and T. Frith, Phys. Rev. A 95, 010102(R) (2017).

\bibitem{SIGMAdva}
M. Znojil,
SIGMA 4,
001 (2008)
(arXiv overlay: math-ph/0710.4432v3).

\bibitem{fabio}
T. Curtright and L. Mezincescu, J. Math. Phys. 48, 092106 (2007);


F. Bagarello, A. Inoue and C. Trapani,
J. Math. Phys. 55, 033501 (2014).

\bibitem{Fluegge}
S. Fl\"{u}gge, Practical Quantum Mechanics II. Springer, Berlin,
1971.

\bibitem{FV}
H. Feshbach and F. Villars, Rev. Mod. Phys. 30, 24 (1958).

\bibitem{KGja}
M. Znojil, J. Phys. A: Math. Gen. 37, 9557 (2004);
%

M. Znojil, 
Czech. J. Phys. 55, 1187 (2005).


\bibitem{PW}
W. Pauli and V. Weisskopf, Helv. Phys. Acta 7, 709 (1934).

\bibitem{Const}
F. Constantinescu and E. Magyari, Problems in Quantum Mechanics.
Pergamon, Oxford, 1971.

\bibitem{[150]}
A. Mostafazadeh, Class. Quant. Grav. 20, 155 (2003).

\bibitem{[151]}
A. Mostafazadeh, Ann. Phys. (N.Y.) 309, 1 (2004).

\bibitem{Dieudonne}
J. Dieudonn\'{e},
in Proc. Int. Symp. Lin. Spaces. Oxford, Pergamon, 1961, p. 115.

\bibitem{Samsonov2}
A. Guo et al, Phys. Rev. Lett. 103, 093902 (2009).

\bibitem{Samsonov}
N. Moiseyev, Non-Hermitian Quantum Mechanics. Cambridge University
Press, Cambridge, 2011;

H. Eleuch and I. Rotter, Phys. Rev. A 95, 022117 (2017).

\bibitem{Jones}
H. F. Jones, Phys. Rev. D 76, 125003 (2007);

M. Znojil, Phys. Rev. D 78, 025026 (2008);

H. F. Jones, Phys. Rev. D 78, 065032 (2008).



\end{thebibliography}
 \end{document}